\documentclass[a4paper,11pt]{article}
\pdfoutput=1

\usepackage[margin=0.8in]{geometry}
\usepackage{array}
\usepackage{bigstrut}

\usepackage{algorithm}
\usepackage{algpseudocodex}

\usepackage{bbm}
\usepackage[utf8]{inputenc}
\usepackage[english]{babel}
\usepackage[T1]{fontenc}
\usepackage{amsmath,amssymb,amsthm}
\usepackage{mathtools}

\usepackage[font=small,labelfont=bf]{caption}
\usepackage{float}
\usepackage{tikz}
\usetikzlibrary{calc,patterns,arrows,decorations.pathreplacing}

\usepackage[font={small,sf}, labelfont=bf]{caption}

\usepackage{setspace}
\usepackage[export]{adjustbox}
\usepackage{makecell}
\usepackage[inline]{enumitem}
\usepackage{url}

\usepackage{xcolor}
\definecolor{darkblue}{RGB}{0,0,128}
\definecolor{darkgreen}{RGB}{0,150,0}

\usepackage[numbers,sort&compress]{natbib}
\usepackage[pdfusetitle]{hyperref}
\hypersetup{breaklinks, colorlinks, linkcolor=blue, citecolor=magenta, filecolor=red, urlcolor=darkblue}

\usepackage{thmtools}
\usepackage{thm-restate}

\newtheorem{theorem}{Theorem}
\newtheorem*{theorem*}{Theorem}
\newtheorem{lemma}[theorem]{Lemma}

\newtheorem{claim}[theorem]{Claim}
\newtheorem{corollary}[theorem]{Corollary}

\newtheorem{definition}[theorem]{Definition}

\usepackage{classical_mystyle}
\graphicspath{{../}{./}{src_tikz/}}

\usepackage{authblk}

\title{Tradeoffs on the volume of fault-tolerant circuits}
\author[1]{Anirudh Krishna}
\author[2,3]{Gilles Z\'emor}
\affil[1]{\it {\small {IBM Quantum, T.\ J.\ Watson Research Center, Yorktown Heights, New York 10598, USA}}}
\affil[2]{\it {\small {Institut de Math\'ematiques de Bordeaux, UMR 5251, Universit\'e de Bordeaux, France}}}
\affil[3]{\it {\small {Institut Universitaire de France}}}
\date{}

\begin{document}
\maketitle

\begin{abstract}
    Dating back to the seminal work of von Neumann \cite{von1956probabilistic}, it is known that error correcting codes can overcome faulty circuit components to enable robust computation.
    Choosing an appropriate code is non-trivial as it must balance several requirements.
    Increasing the rate of the code reduces the relative number of redundant bits used in the fault-tolerant circuit, while increasing the distance of the code ensures robustness against faults.
    If the rate and distance were the only concerns, we could use asymptotically optimal codes as is done in communication settings.
    However, choosing a code for \emph{computation} is challenging due to an additional requirement: 
    The code needs to facilitate accessibility of encoded information to enable computation on encoded data.
    This seems to conflict with having large rate and distance.
    We prove that this is indeed the case, namely that a code family cannot simultaneously have constant rate, growing distance and short-depth gadgets to perform encoded $\CNOT$ gates.
    As a consequence, achieving good rate and distance may necessarily entail accepting very deep circuits, an undesirable trade-off in certain architectures and applications.
\end{abstract}

\section{Introduction}

Error correcting codes, beyond their extensive use in communication and storage, can be employed to ensure reliable computation in the presence of noise \cite{von1956probabilistic,dobrushin1977lower,dobrushin1977upper,pippenger1985networks,pippenger1991lower,gacs1994lower,spielman1996highly,romashchenko2006reliable}.
In this setting, one can use an error correcting code $\ccode$ to simulate a circuit $\cid$ using a fault-tolerant circuit $\cft$.
There is an intimate relationship between the choice of error correcting code and the size of the resulting circuit $\cft$.
A fundamental question in fault-tolerant circuit design is: What is the smallest fault-tolerant circuit $\cft$ that can be used to simulate the noisy circuit $\cid$ up to some desired degree of robustness?

The problem of constructing reliable circuits dates back to von Neumann \cite{von1956probabilistic}.
In the von Neumann fault model, each elementary gate in the circuit $\cft$ can fail with some constant probability and we desire its output to be correct with a sufficiently large probability.
To overcome this problem, von Neumann and works that followed used error correcting codes to encode information and, in parallel, designed fault-tolerant gadgets to perform a universal set of operations on encoded data.
These results prove that an ideal circuit $\cid$ with volume $\vol$ can be replaced by a fault-tolerant circuit $\cft$ with volume $\volft = O(\vol \, \log(\vol))$.
We review further developments in Sec.~\ref{subsec:related-work}.
In contrast to these achievability results, our focus is on understanding \emph{lower bounds} on the efficiency of the fault-tolerant circuit $\cft$.

\textbf{Main question:} Can we construct a circuit $\cft$ that obeys $\volft = O(\vol)$ while simultaneously being robust to many errors?
We consider a setting, to the best of our knowledge, first explicitly considered by Spielman \cite{spielman1996highly} \footnote{See \cite{spielman1996highly} Section 3; specifically, the discussion following Remark 4.}.
In this setting, we compute directly on encoded information and allow the output to be a corrupt codeword.
We disregard the task of encoding to/ decoding from the error correcting code.
Secondly, as Spielman highlights, we would like to choose an error correcting code that can be used to simulate any circuit $\cid$.
If we choose $\ccode$ for a specific computation, we might find a trivial solution where the computation to be performed is hidden in the encoding map for the code $\ccode$.

\textbf{Our main result: A lower bound on the size of fault-tolerant circuits}
We provide evidence that if the volume $\volft$ of the fault-tolerant circuit $\cft$ obeys $\volft = O(\volid)$, where $\volid$ is the volume of the circuit $\cid$ that we wish to simulate, then the fault-tolerant circuit $\cft$ can only tolerate a bounded number of errors.
To be precise, if we use a binary linear code $\ccode: \bbF_2^{k} \to \bbF_2^{n}$ to construct the fault-tolerant circuit $\cft$, then the code $\ccode$ cannot have good distance.

Intuitively, we \emph{expect} this problem to be difficult as it requires conflicting constraints.
For example, consider a binary linear code $\ccode$ with large distance.
The codewords must be packed together densely to achieve good rate.
This complicates encoded computation as, to address a \emph{subset} of encoded bits, we expect to first have to untangle codewords.
Otherwise, we may accidentally affect codewords that share support with the codeword we wish to address.
Our result formalizes this intuition.

\textbf{Additional context:} This work was originally motivated by studying \emph{quantum} error correcting codes which--like their classical counterparts--can be used to robustly simulate \emph{quantum} circuits.
This foundational idea was established in a series of landmark results on the subject \cite{aharonov1997fault,knill1996concatenated,kitaev1997quantum,aliferis2005quantum}.
Physical components in quantum circuits fail at rates far too high to reliably implement algorithms of interest without some form of quantum error correction\footnote{For example, see \cite{beverland2021cost} for estimates of circuit sizes required for implementing algorithms of interest.}.
Designing efficient fault-tolerant quantum circuits thus remains a central challenge in quantum error correction, viewed both through an engineering lens and also a more fundamental point-of-view where one seeks to understand asymptotic tradeoffs.
Progress in the asymptotic regime has highlighted the relationship between the choice of error correcting code and the size of the resulting fault-tolerant circuit \cite{kovalev2013fault,gottesman2013fault,yamasaki2024time}.
This has then spurred paradigm changes, such as the push toward quantum low-density parity-check (LDPC) codes (See \cite{gottesman2013fault}).
Quantum LDPC codes, by virtue of their rate, can optimize the total number of physical qubits needed to implement a fault-tolerant circuit.
This is especially valuable given the persistent difficulty of engineering large numbers of reliable, well-connected qubits.

However, the advantages of quantum LDPC codes do not guarantee that they are universally the best choice.
We need to account for the \emph{volume} of the fault-tolerant circuit, which includes the total time required to implement it, and not merely the total number of physical qubits used at any single point in time.
This, in turn, means that the quantum error correcting code used must allow for efficient ways to perform fault-tolerant encoded operations.
Indeed, Gottesman's construction \cite{gottesman2013fault} works with simulating sparse circuits $\cid$, i.e.\ in each time step, only a limited number of encoded operations are performed.
It is unclear whether the applicability of quantum LDPC codes extends to circuits $\cid$ that are not sparse.
Thus, to assess the suitability of quantum LDPC codes for a problem of interest, further progress needs to be made on performing fault-tolerant encoded operations.

Our work contributes to this broader effort by providing evidence of inherent limits: Even when codes have good rate, there may be fundamental constraints on how much they can reduce the volume of fault-tolerant circuits.

\subsection{Technical Overview}
\label{subsec:tech-overview}

\subsubsection{Our Model}

We consider circuits that process binary strings, constructed using a finite set of gates with fan-in and fan-out at most 2.
Let $\cid$ be a circuit with width $\widthid$ and depth $\depthid$, and let $\cft$ be the circuit that simulates $\cid$ with width $\widthft$ and depth $\depthft$.
The width and depth are figures-of-merit to assess the size of circuits:
The \emph{depth} $\depthid$ ($\depthft$) of the circuit $\cid$ ($\cft$) is the total number of time steps it takes to execute the circuit.
This can be non-trivial as each bit can only participate in 1 gate in each time step.
The \emph{width} $\widthid$ ($\widthft$) of the circuit $\cid$ ($\cft$) is the total number of bits it uses.
This includes any ancillary bits used as scratch space during the execution of the circuit.
The \emph{space overhead} is the ratio $\widthft/\widthid$.
The volume $\vol$ ($\volft$) is the product of the width $\width$ ($\widthft$) and the depth $\depth$ ($\depthft$).
The \emph{volume overhead} is the ratio $\volft/\volid$.
Our aim is to understand whether we can achieve constant volume overhead in a specific model that we now explain.

We consider the class of \emph{sparse} circuits $\cid$: In each time step $1 \leq t \leq \depthid$ of $\cid$, there is exactly one gate $\elemid(t)$ between two registers.
To be clear, the circuit $\cft$ itself and the gadgets $\elemft$ need not be sparse circuits; the sparsity constraint only applies to $\cid$.
Further simplifying our model, we assume that for all $1 \leq t \leq \depthid$, the gate $\elemid(t)$ is a $\CNOT$ gate\footnote{\textbf{Remark:} We study the $\CNOT$ gate rather than the ${\rm XOR}$ gate that appears more commonly as part of a universal gate set in classical computation.
Note, however, that we are simply using a reversible ${\rm XOR}$ gate as the action of the $\CNOT_{i,j}$ gate on $\m \in \bbF_2^{k}$ maps $(m_i,m_j) \mapsto (m_i, {\rm XOR}(m_i,m_j))$.}.
To be precise, there exist some distinct indices $i_t, j_t \in [k]$ such that $\elemid(t)$ is the gate $\CNOT_{i_t,j_t}$.
We say that $\elemft$ is a \emph{targeted} gate gadget because we can target the pair $i_t,j_t \in [k]$ while not affecting the remaining encoded bits.

Let $\m \in \bbF_2^{k}$ correspond to the input to $\cid$ and $\ccode: \bbF_2^{k} \to \bbF_2^{n}$ an error correcting code.
We assume that the encoded input $\ccode(\m)$ is given to us, i.e.\ we do not need to encode $\m$ into the code $\ccode$, nor do we decode the final result from the code.
If the width of the circuit $\widthid < k$, then we assume that $\m$ is padded by $0$s to make it the appropriate length.
We restrict our attention to linear codes as non-linear codes are often unwieldy in practice.
To construct $\cft$, we consider a step-wise process: For all time steps $1 \leq t \leq \depthid$, every gate $\elemid(t)$ in the circuit $\cid$ is replaced by a \emph{gadget} $\elemft(t)$.
The gadget $\elemft(t)$ itself a circuit that performs the gate $\elemid(t)$ directly on encoded information, i.e.\ it obeys
\begin{align}
\label{eq:def-encoded-op}
    \elemft(t) \circ \ccode = \ccode \circ \elemid(t)~.
\end{align}
For all $t \in \depthid$, we perform a round of error correction after the gadget $\elemft(t)$ (see schematic in Fig.~\ref{fig:circuit-intro}).
We explain our error model below.
By correcting errors at regular intervals, we guarantee that the number of corruptions does not overwhelm the error correcting code $\ccode$.

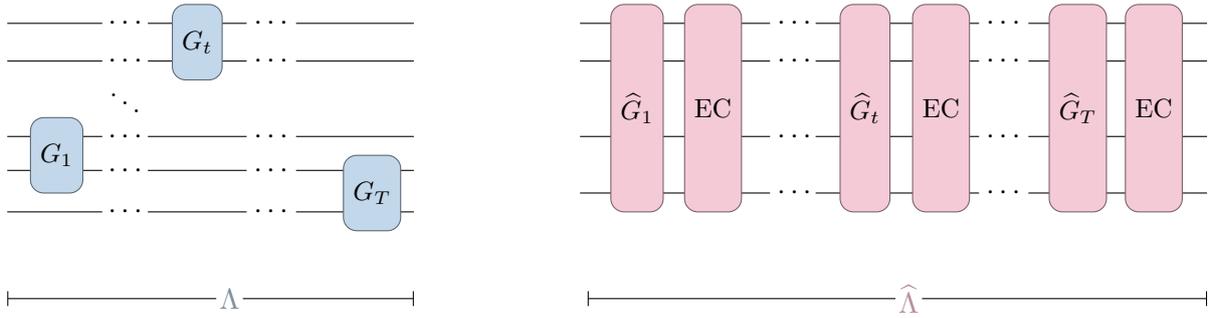
\begin{figure}[h]
    \centering
    \usetikzlibrary{positioning}
\definecolor{pastelblue}{RGB}{179,205,227}
\definecolor{pastelpink}{RGB}{241,189,205}

\colorlet{pastelblueborder}{pastelblue!50!black}
\colorlet{pastelpinkborder}{pastelpink!50!black}

\begin{tikzpicture}[gate/.style={draw=pastelblueborder, text=black, minimum width=0.5cm, minimum height=0.5cm},
                    dgate/.style={draw=pastelblueborder, minimum width=0.5cm, minimum height=1cm},
                    gadget/.style={draw=pastelpinkborder, minimum width=0.5cm, minimum height=2.75cm}]

  \foreach \y in {2.25,1.75,0.75,0.3,-0.25} {
    \draw (-3.5,\y) -- (-2.25,\y) node[xshift=9.5] {$\cdots$} (-1.65,\y)--(-0.35,\y) node[xshift=10] {$\cdots$} (0.3,\y)--(1.85,\y);
  }
  \node at (-1.95,1.3) {$\ddots$};
  
  \node[dgate,fill=pastelblue!80,rounded corners=5pt] at (1.3,0) {{\small $\elemid_{\depthid}$}};
  \node[dgate,fill=pastelblue!80,rounded corners=5pt] at (-1,2) {{\small $\elemid_t$}};
  \node[dgate,fill=pastelblue!80,rounded corners=5pt] at (-2.85,0.5) {{\small $\elemid_1$}};

    \draw[|-] (-3.5,-1.4) -- (-0.75,-1.4) node[xshift=5] {{\color{pastelblue!70!black}$\cid$}};
    \draw[-|] (-0.4,-1.4)--(1.85,-1.4);

\begin{scope}[xshift=250pt]
  \foreach \y in {2.25,1.75,0.75,0} {
    \draw (-4.75,\y) -- (-2.25,\y) node[xshift=10] {$\cdots$} (-1.65,\y)--(0.5,\y) node[xshift=10] {$\cdots$} (1.15,\y)--(3.5,\y);
  }
  \node at (-3.25,1.5) {$\vdots$};
  
  \node[gadget,fill=pastelpink!80,rounded corners=5pt] at (1.8,1.125) {{\small $\elemft_{\depthid}$}};
  \node[gadget,fill=pastelpink!80,rounded corners=5pt] at (2.8,1.125) {{\small ${\rm EC}$}};
  \node[gadget,fill=pastelpink!80,rounded corners=5pt] at (-1,1.125) {{\small $\elemft_t$}};
  \node[gadget,fill=pastelpink!80,rounded corners=5pt] at (0,1.125) {{\small ${\rm EC}$}};
  \node[gadget,fill=pastelpink!80,rounded corners=5pt] at (-3,1.125) {{\small ${\rm EC}$}};
  \node[gadget,fill=pastelpink!80,rounded corners=5pt] at (-4,1.125) {{\small $\elemft_1$}};

  \draw[|-] (-4.65,-1.4) -- (-0.6,-1.4) node[xshift=5] {{\color{pastelpink!75!black} $\cft$}};
    \draw[-|] (-0.25,-1.4)--(3.5,-1.4);
\end{scope}

\end{tikzpicture}
    \caption{Schematic illustrating the conversion of a sparse circuit $\cid$, on the left, into a fault-tolerant circuit $\cft$, on the right.
    In the $t$\textsuperscript{th} time step in $\cid$, there is one elementary gate $\elemid_t$ that acts on at most $2$ bits.
    In $\cft$, the gate $\elemid_t$ is simulated using the gadget $\elemft_t$ that acts directly on encoded information.
    It is then followed by a round of error correction.}
    \label{fig:circuit-intro}
\end{figure}

We consider a simple adversarial error model, where, for $1 \leq t \leq \depthid$, an adversary is allowed to introduce erasure errors prior to each gate $\elemft(t)$.
To be clear, the adversary is allowed to corrupt bits that are not in the support of the gadget $\elemft(t)$.
We assume that the gates themselves are free of faults; this assumption is justified below.
However, while the gates are free of faults, even ideal gates can \emph{spread} erasures: If a gate acts on even one input bit that is erased, then all output bits are considered erased.
For $E \subseteq [n]$, we use $\ccode(\m)\rvert_{E^c}$ to denote the restriction of the codeword to the coordinates $E^c = [n] \setminus E$; the bits in locations corresponding to $E$ have been replaced by the symbol $\perp$ denoting corruption.
The erasure map itself is denoted $\era_{E}: \{0,1\}^{n} \to \{0,1,\perp\}^n$.

Henceforth, we refer to gadgets $\elemft$ corresponding to $\elemid$ without an explicit time index $t$, with the assumption that the gadgets appear somewhere in the circuits $\cft$ and $\cid$ respectively.
Let $\elemft$ be a gadget corresponding to a gate $\elemid$, i.e.\ it obeys $\elemft \circ \ccode = \ccode \circ \elemid$.
Let $\ccode$ have distance $d$.
For $\epsilon \leq (0,1)$, we say $\elemft$ is $\epsilon$-robust if, for any erasure $E \subseteq [n]$ of weight $|E| \leq \epsilon d$, the output of $\elemft$ is correctable by an ideal decoder.
To be precise, we say $\elemft$ is $\epsilon$-robust if for all $E \subseteq n$ satisfying $|E| \leq \epsilon d$, there exists an erasure $F \subseteq [n]$ such that
\begin{align}
\label{eq:robustness}
    \elemft \circ \era_{E} \circ \ccode  = \era_{F} \circ \elemft \circ \ccode~,
\end{align}
and $|F| < d$.

In this paper we constrain gadgets in order to achieve constant volume overhead: We require that the width $\widthft$ of the circuit $\cft$, the total number of registers it uses, including any ancillary registers used as scratch space, is at most $O(n)$.
In particular, for each gadget $\elemft$ in the circuit $\cft$, we allow for at most a $O(n)$ ancillary bits.
Secondly, we stipulate that the depth of all gadgets $\elemft$ is a constant, independent of the width $\widthid$.
This is for two reasons.
\begin{enumerate}
    \item Firstly, we wish to understand the extent to which we can reduce the volume of the circuit $\cft$.
    If the gadget $\elemft$ has constant depth and uses only $O(n)$ ancillary bits, then we can hope to achieve $\volft = O(\volid)$.
    \item Secondly, short-depth circuits are automatically robust.
    To be precise, consider an input codeword $\w \in \ccode$ corrupted by an erasure $E \subseteq [n]$.
    After the gadget $\elemft$, suppose the erasure $E$ becomes an erasure $F \subseteq [n]$.
    If $E$ has weight $w$, and the gadget $\elemft$ has depth $\ell = O(1)$, then $F$ can have weight at most $2^{\ell} w$.
    Consequently, as long as $\elemft$ satisfies Eq.~\eqref{eq:def-encoded-op}, it is automatically $\epsilon$-robust where $\epsilon = 2^{-\ell}$.
\end{enumerate}
With this, we can return to the powers granted to the adversary.
Allowing them to introduce erasure errors in each time step of the gadget $\elemft$ will only change the robustness $\epsilon$ by a constant factor.

Although the time required to perform error correction must technically be included in the depth $\depthft$ of the circuit $\cft$, we ignore it.
To justify this, we note that the volume of circuits required to perform decoding has already been studied under various models \cite{bar2002streaming,gronemeier2006note}.
These lower bounds on the volume apply even to memories, i.e.\ in settings where we merely store information and retrieve it later.
In contrast, we want to focus on lower bounds on the volume of the circuit that emerge from the need to perform \emph{encoded operations}.
We cannot, however, cheat to use the unbounded resources granted for error correction operation to perform the gate itself.
In particular, Eq.~\eqref{eq:robustness} demands that the output of the gadget $\elemft$ should be within a ball of radius at most $d$ from the correct codeword.
This is sufficient for us to prove our main result.

As additional motivation, we remark that when implementing quantum circuits, the decoding operations are performed by classical computers.
When assessing the cost of constructing a fault-tolerant quantum computer, we only wish to account for the depth of \emph{quantum} computation, i.e.\ we only tally the depth required to implement the gate $\elemft$ using quantum gates.
We disregard polynomial-time classical computation used in decoding the error correcting code.

\subsubsection{Our Main Result}

Having established our setup, we proceed to state our main result and provide some intuition for the proof technique.

Theorem~\ref{thm:main} stated below highlights an intimate relationship between the code $\ccode$ and the volume overhead of the resulting fault-tolerant circuit.

\begin{restatable}[]{theorem}{main}
\label{thm:main}
Consider an infinite family of $[n,k(n),d(n)]$ codes $\famC = \{\ccode : \bbF_2^{k} \to \bbF_2^{n}\}$ such that for all codes $\ccode : \bbF_2^{k} \to \bbF_2^{n}$, for all $i \in [k]$, there exists $j \in [k]$, $j \neq i$, and a gadget $\elemft_i$ such that
\begin{enumerate}
    \item $\elemft_i$ encodes the gate $G_i=\CNOT_{i,j}$, i.e.\ $\elemft_{i} \circ \ccode = \ccode \circ \elemid_{i}$. 
    \item $\elemft_i$ has constant depth $\ell$.
    \end{enumerate}
    Then we have the relationship
    \begin{align*}
        k = O\left(\frac{n}{d^{1/q}} \right)~.
    \end{align*}
    where $1 < q \leq 2^{\ell}$ is a constant.
\end{restatable}

Before proceeding to the proof idea, we make some remarks to add some context to the assumptions and conclusion of the theorem statement.
\begin{enumerate}
\item In the rest of the paper, we will make the assumption $kd = \omega(n)$.
This is to ensure that we consider codes that are better than the repetition code which saturate the bound $kd = O(n)$.
Performing targeted gates on the repetition code is simple: each bit is encoded in a separate block.
Therefore, we can perform the desired operation on encoded bits by performing the corresponding operation in parallel on each bit of the error correcting code.
\item 
Let $B_{\ccode} = \{\g_1,...,\g_k\}$ be a basis for the code $\ccode$.
Consider a gadget $\elemft_i$ that corresponds to a $\CNOT$ gate $\elemid = \CNOT_{i,j}$ between a pair of distinct bits $i,j \in [k]$.
Note that if $\elemft_i$ is implemented in constant depth $\ell$, then we must have $|\g_j| \leq 2^{\ell} |\g_i|$.
We assume that this is true for all $i \in [k]$.

\item Finally, we discuss how the rate of the code $\ccode : \bbF_2^{k} \to \bbF_2^{n}$ relates to the volume $\volft$ of the fault-tolerant circuit $\cft$.
Note that the code $\ccode$ must be chosen such that $k \geq \widthid$; otherwise, we would not be able to encode the input $\s \in \bbF_2^{\widthid}$ to the circuit $\cid$.
Second, the fault-tolerant circuit $\cft$ must at least use the bits of the error correcting code, and therefore, $\widthft \geq n$.
Together, these observations imply that the space overhead $\widthft/\widthid$ obeys
\begin{align}
\label{eq:space-ovhd-bound}
    \frac{\widthft}{\widthid} \geq \frac{n}{k}~,
\end{align}
i.e.\ it is lower bounded by the inverse of the rate of the code $\ccode$.
Therefore, if we can show that the code does not have constant rate under certain conditions, it implies that the space overhead must grow in an unbounded manner.  
\end{enumerate}

\textbf{Proof idea:}
First, we observe that the statement of Theorem~\ref{thm:main} does not require the existence of a gadget to perform \emph{every} possible encoded $\CNOT$ gate on the code $\ccode$.
It merely requires that for all $i \in [k]$, there exists a gadget to perform a $\CNOT$ gate that uses $i$ as the control.

Using this observation, we show that if the constant-depth gadgets required by Theorem~\ref{thm:main} were to exist, then the code $\ccode$ must be equivalent to what we refer to as a $(q,r)$ local code.
We say that $\ccode$ is a $(q,r)$ local code if for all $\m \in \bbF_2^{k}$ and all $i \in [k]$, we can infer $m_i$ by querying $\ccode(\m)$ in only $q$ locations.
Furthermore, there are $r$ disjoint such sets of $q$ locations to infer $m_i$.
If the gadgets required by Theorem~\ref{thm:main} exist, then $\ccode$ is a $(q,r)$ local code where $q = O(1)$ and $r = \Theta(d)$.
These are generalizations of well-studied combinatorial objects called locally decodable codes (where $r = \Omega(n)$).
It is known that $(q,r)$ local codes must have poor rate: A $(q,r)$ local code must obey
\begin{align}
    k = O\left(\frac{n}{r^{1/q}} \right)~.
\end{align}

We conclude by recalling Eq.~\eqref{eq:space-ovhd-bound}, which stated that the space overhead was bounded by the inverse of the rate of the code:
\begin{align*}
    \frac{\widthft}{\widthid} \geq \frac{n}{k}~.
\end{align*}
If we want to construct fault-tolerant circuits $\cft$ that are robust to $\omega(1)$ adversarial erasure errors, then we require $d = \omega(1)$.
However, Theorem~\ref{thm:main} shows that we cannot simultaneously have constant space overhead and $d = \omega(1)$.

Our work implies that if a code supports targeted $\CNOT$ gates then it must contain a $(q,r)$ local code.
On the other hand, if the code is a $(q,r)$ local code, then it is also possible to construct explicit, constant-depth gadgets to perform targeted $\CNOT$ gates.
For the sake of completeness, we include a brief description in Appendix~\ref{app:explicit}.
We note that a much more general construction for fault-tolerant computation using locally decodable codes was already presented by Romaschenko \cite{romashchenko2006reliable} and is described in Sec.~\ref{subsec:related-work}.

\subsection{Related work}
\label{subsec:related-work}

As stated in the introduction, the construction of fault-tolerant circuits goes back to von Neumann.
The von Neumann fault model considers \emph{stochastic} errors: Each elementary gate in the circuit $\cft$ can fail with some constant probability $p$ and we desire its output to be correct with some sufficiently small probability $\widehat{p}$.
To overcome this problem, each bit in the ideal circuit is encoded locally using a repetition code; the value of the bit is obtained via majority vote.
The one- (two-)bit encoded gate $\elemft$ is implemented by performing the gate $\elemid$ in parallel on each (pair of) bits in the corresponding repetetion code(s).
This result, and those that build on it (Dobrushin \& Ortyukov \cite{dobrushin1977upper} and Pippenger \cite{pippenger1985networks}), show that an ideal circuit $\cid$ with volume $\vol$ can be replaced by a fault-tolerant circuit $\cft$ with volume $\volft = O(\vol \, \log(\vol))$.

Spielman \cite{spielman1996highly} also considered computing over encoded information.
He showed that it is possible to construct a fault-tolerant circuit $\cft$ with width $\widthft = O(\widthid\, \log^{O(1)}\widthid)$ and depth $\depthft = O(\depthid \, \log^{O(1)} \widthid)$.
This is accomplished using generalized Reed-Solomon codes
\footnote{Presciently, Spielman asks whether such ideas can be used to `compensate for decoherence in quantum computations' in Section 9 of \cite{spielman1996highly} on further directions.}.
Closer to the subject of this paper, Romaschenko \cite{romashchenko2006reliable} demonstrated that it is possible to construct fault-tolerant circuits using locally decodable codes.
On the achievability front, there has also been work on understanding the capacity of computation under various constraints \cite{simon2011capacity,grover2014shannon,yang2017computing}.

Lower bounds on the space overhead for this model of coded computation have been studied \cite{peterson1959codes,winograd1962coding,rachlin2008framework}.
Rachlin \& Savage \cite{rachlin2008framework} consider an $[n,k(n),d(n)]$ code family $\famC$ and study how to perform certain component-wise operations.
To be precise, consider $\ccode: \bbF_2^{k} \to \bbF_2^{n}$ in the family $\famC$, consider codewords $\ccode(\m_1), \ccode(\m_2)$ for $\m_1,\m_2 \in \bbF_2^{k}$.
Suppose we can perform the operation $T$ such that $T(\ccode(\m_1),\ccode(\m_2)) \mapsto \ccode(\m_1 \star \m_2)$, where $\star$ represents the component-wise product of $\m_1$ and $\m_2$.
If each output location of $T$ depends on at most $q$ input locations, then they show $k \leq qn/d$.
In contrast, we consider the setting where we perform encoded operations within a \emph{single} block of a code.
We may wish to do so in settings where we want to improve the space overhead $\widthft/\widthid$ as expressing a code $\ccode$ as a direct sum of two codes $\ccode_1,\ccode_2$ cannot improve the rate.

As stated in the introduction, we only consider settings where the input and output are both encoded.
We disregard the complexity of encoding to/ decoding from the error correcting code.
There exist results that include these operations and bound the volume of the resulting circuit $\cft$.
This approach uses the so-called \emph{sensitivity} of a function $f$, a measure of how many bits are needed for the function $f$ to evaluate the output starting from an input that is \emph{not} encoded.
Computing a function $f$ with sensitivity $s$ requires a circuit with volume $\Omega(s\, \log(s))$ \cite{dobrushin1977lower, pippenger1991lower, gacs1994lower,reischuk1991reliable,evans1995information}.

\textbf{Quantum setting:} In an elegant work connecting quantum codes and automorphism groups, Guyot and Jaques \cite{guyot2025addressability} recently proved that circuits composed of only single-qubit Clifford gates cannot be used to implement targeted encoded $\CNOT$ gates unless $kd = O(n)$ (Ref.\ Corollary 1 in \cite{guyot2025addressability}).
This result also demonstrates that there are limits to performing other encoded Clifford gates on quantum codes such as the Hadamard and phase gates.

\subsection{Acknowledgements}

We would like to thank Nou\'edyn Baspin, Micheal Beverland, Andrew Cross, Louis Golowich, Venkat Guruswami, Ray Li, Chris Pattison, Mary Wootters and Ted Yoder for discussions and feedback.
Part of this work was done while the authors were visiting the Simons Institute for the Theory of Computing, supported by NSF QLCI Grant No. 2016245. The second author was also supported by Plan France 2030 through the project NISQ2LSQ, ANR-22-PETQ-0006.

\section{Background \& Notation}
\label{sec:background}

\subsection{Basic Definitions}
\label{subsec:background-classical}

A binary, $[n,k,d]$ linear code $\ccode : \{0,1\}^{k} \to \{0,1\}^n$ is a linear map from the space of $k$-bit strings to the space of $n$-bit strings.
We refer to the privileged basis for the code space ${\rm im}(\ccode)$ via the set $\{\g_1,...,\g_k\}$ where for $i \in [k]$, $\g_i$ is the image of the standard basis vector $\e_i \in \bbF_2^{k}$.
Alternatively, we can think of the code as a specification of a $k$-dimensional subspace of $\bbF_2^{n}$ (rather than an encoding map) together with a basis.
When the meaning is unambiguous, we refer to $\ccode$, instead of its image ${\rm im}(\ccode)$, as the code space.
The distance $d$ of the code is the minimum weight of a non-zero element of the code space.
With this notation, $\ccode^{\perp}$ represents the space that is dual to the code space $\ccode$, i.e.\ the set of vectors $\u \in \FF_2^{n}$ such that $\langle \u, \g_{i} \rangle = 0 \pmod{2}$ for all $i \in [k]$.
We study asymptotic properties of an infinite code \emph{family}, i.e.\ we consider $\famC = \{\ccode_i: \FF_2^{k_i} \to \FF_2^{n_i} \,,\, i \in \bbN\}$ where for all $i \in \bbN$, $n_{i+1} > n_{i}$ and we study the asymptotic rate $\lim_{i \to \infty} k_i/n_i$ and asymptotic relative distance $\lim_{i \to \infty} d_i/n_i$.
We use an implicit parameterization of $k$ and $d$ and simply express $k = k(n)$ and $d = d(n)$ as functions of the code size $n$ and say that $\famC$ is a $[n,k(n),d(n)]$ code family.

For $E \subseteq [n]$, we let $\era_{E}: \{0,1\}^n \to \{0,1,\perp\}^n$ be the \emph{erasure error} that replaces the symbols in locations corresponding to $E$ by a special symbol $\perp$; it leaves symbols in $E^c$ unchanged.
For any word $\w \in \bbF_2^{n}$, we let $\w\rvert_{E^c}$ denote the $(n-|E|)$-bit vector obtained from $\w$ by restricting its coordinates to $E^c$.

\subsection{Locality}
In this section, we define a $(q,r)$ local code.
This is a straightforward extension of the idea of locally decodable codes studied by Katz \& Trevisan \cite{katz2000efficiency}.
We return to this connection after the definition.
Formally, a $(q,r)$ local code is defined as follows.

\begin{definition}
\label{def:locality}
    Let $\ccode: \FF_2^{k} \to \FF_2^{n}$ be a linear code whose basis is $\g_1,...,\g_{k}$.
    We say $\ccode$ is a $(q,r)$ local code if, for all $i \in  [k]$, there exist a set $U^{i} \subset \{0,1\}^{n}$ of $r$ distinct vectors such that:
    \begin{enumerate}
        \item \label{it:bdd-wt} For all $\u^{i} \in U^{i}$, the Hamming weight of $\u^{i}$ is at most $q$:
        \begin{align*}
            |\supp(\u^{i})| \leq q~.
        \end{align*}
        \item \label{it:pairwise-disjoint} Distinct $\u^{i},\v^{i} \in U^{i}$ have disjoint support:
        \begin{align*}
            \supp(\u^{i}) \cap \supp(\v^{i}) = \emptyset~.
        \end{align*}
        \item \label{it:inner-prod-reln-ldc} For all $\u^{i} \in U^{i}$, the message symbol $m_{i}$ is recovered by adding the corresponding elements of $\ccode(\m)$ in the support of $\u^{i}$, i.e.\ for all $\m \in \bbF_2^{k}$,
        \begin{align*}
           \langle \u^{i}, \ccode(\m)\rangle = m_{i}~.
        \end{align*}
    \end{enumerate}
\end{definition}

We will restrict our attention to $(q,r)$ local codes where $q = O(1)$ and $r = \omega(1)$.
We also remark that if $q = O(1)$ and $r = \Omega(n)$, then a $(q,r)$ local code is a locally decodable code as defined by Katz \& Trevisan \cite{katz2000efficiency}.
Our definition is also related to codes with the disjoint repair group property, batch codes and Private Information Retrieval (PIR) codes; we point the interested reader to the review by \cite{skachek2018batch}.

Our definition allows for settings where the distance $d$ of the code family grows sub-linearly with $n$.
In the regime $r = \Omega(n)$, there are several equivalent ways of defining a locally decodable code.
We do not delve into these variations; the definition of locally decodable codes above is most suitable for our purposes.
For a broad overview of this subject, we point the interested reader to the review by Yekhanin \cite{yekhanin2012locally}.

We highlight some properties of $(q,r)$ local codes that are straightforward consequences of Definition~\ref{def:locality} that the reader may find helpful to keep in mind.
\begin{enumerate}
    \item \textbf{Basis dependence:} A $(q,r)$ local code is defined with respect to a fixed basis $\g_1,...,\g_k$ for the code $\ccode$.
    For other bases, the vectors in $U^{i}$ need not have bounded weight $q$ (Def.~\ref{def:locality}, Property~\ref{it:bdd-wt}) nor need they remain pairwise disjoint (Def.~\ref{def:locality}, Property~\ref{it:pairwise-disjoint}).
    \item \textbf{Redundancy $r$ versus code distance $d$:}
    Let $\ccode: \FF_2^{k} \to \FF_2^{n}$ be an $[n,k,d]$ code.
    If $\ccode$ is a $(q,r)$ local code, then $d = \Omega(r)$.
    This guarantee does not automatically extend to non-linear codes \cite{cheraghchi2012correctness}.
\end{enumerate}

It is known that $(q,r)$ local codes are severely limited in their ability to encode information \cite{katz2000efficiency,deshpande2002better,wehner2005improved,yekhanin2012locally}.

\begin{restatable}{theorem}{boundldc}
\label{thm:bound-ldc}
    Let $\famC$ be a $[n,k(n),d(n)]$ family of $(q,r)$-local codes such that $q \geq 2$ is a constant and $r = \Omega(d)$.
    Then
    \begin{align*}
        k(n) = O\left( \frac{n}{d(n)^{1/q}} \right)~.
    \end{align*}
\end{restatable}
The bound in Theorem~\ref{thm:bound-ldc} thus limits the rate of the code if the distance $d(n) = \omega(1)$.
For a proof of this statement, see Theorem 5.4 from \cite{yekhanin2012locally} \footnote{The proof in Yekhanin's review \cite{yekhanin2008towards} is for the setting $r = \Omega(n)$ but it is easily seen to extend to the relaxations where $r = \omega(1)$.}.
It remains an open problem to find codes that saturate these bounds for $q = O(1)$ and $r = \Omega(n)$.

\subsection{Gates \& Circuits}

A circuit is composed of a series of $1$- and $2$-bit gates, i.e.\ the fan-in and fan-out of these gates are at most $2$.
We assume that erasure errors propagate in the worst possible manner: When one input of a gate is erased, then all of its output bits are assumed to be erased.

\begin{definition}
\label{def:encoded-gate}
    Consider a code $\ccode: \bbF_2^{k} \to \bbF_2^{n}$.
    Let $\elemid: \bbF_2^{k} \to \bbF_2^{k}$ be an operation on $k$ bits.
    We say that a circuit $\elemft : \bbF_2^{n} \to \bbF_2^{n}$ is the gadget corresponding to $\elemid$ if it obeys
    \begin{align*}
        \elemft \circ \ccode = \ccode \circ \elemid~.
    \end{align*}
    \end{definition}
Note that such a gadget $\elemft$ corresponding to $\elemid$ need not be unique.

Next, we formally define what it means for a gadget to be robust.

\begin{definition}
\label{def:robust}
Let $d$ be the distance of the code $\ccode$.
For a constant $\epsilon \in (0,1)$, a gadget $\elemft$ corresponding to $\elemid$ is $\epsilon$-robust if for all $E \subseteq [n]$ such that $|E| < \epsilon d $,
\begin{align*}
    \elemft \circ \era_{E} \circ \ccode = \era_{F} \circ \elemft \circ \ccode~,
\end{align*}
where $F \subseteq [n]$ is an erasure pattern that is decodable by an ideal decoder.
\end{definition}

We will focus on the $\CNOT$ gate whose action on $k$ bits can be defined as follows: let $\{\e_i\}_{i \in [k]}$ be the standard basis of $\bbF_2^{k}$.
For all $\m \in \bbF_2^{k}$, for distinct $i,j \in [k]$, the gate $\CNOT_{i,j}$ with the $i$\textsuperscript{th} bit as \emph{control} and $j$\textsuperscript{th} bit as \emph{target} is defined as
\begin{align}
    \CNOT_{i,j}:\; \m \mapsto \m + m_i\,\e_j~.
\end{align}

Such gadgets can be used to construct a fault-tolerant circuit as described in the Technical Overview (See Sec.~\ref{subsec:tech-overview}).

\section{Main result}

In this section, we shall prove our main result: We show that if a family of classical codes supports a large set of encoded $\CNOT$ gates such that the corresponding circuits can be implemented in short depth, then the code cannot have constant rate.
This is stated formally in our main theorem which we restate here for convenience.

\main*

For the remainder of the paper, we assume that the code family obeys $kd = \omega(n)$.

\subsection{Short-Depth Circuits are Robust to Erasures}

In this section, we prove that that short-depth circuits are naturally $\epsilon$-robust.

Before doing so, we show that Definition~\ref{def:robust} captures the idea that if $\elemft$ is $\epsilon$-robust, it performs the desired operation approximately even when the input is noisy.
This is proved in the following lemma.

\begin{lemma}
    Let $\ccode : \bbF_2^{k} \to \bbF_2^{n}$ be a binary linear code and $\elemft$ be an $\epsilon$-robust gadget corresponding to $\elemid$.
    Let $E \subseteq [n]$ such that $|E| \leq \epsilon d$.
    For all $\m \in \bbF_2^{k}$, if the input $\era_{E} \circ \ccode(\m)$ to the circuit $\elemft$ is partially corrupted, then its output can be recovered using an ideal decoder.
\end{lemma}
\begin{proof}
Consider the action of $\elemft$ on a partially corrupted codeword:
\begin{align}
    \elemft \circ \era_{E} \circ \ccode(\m)
    &= \era_{F} \circ \elemft \circ \ccode(\m)\\
    &= \era_{F} \circ \ccode \circ \elemid (\m)~.
\end{align}
In the first equality, we have used Def.~\ref{def:robust} of an $\epsilon$-robust gate; in the second equality, we have used the definition of the encoded gadget.
Using an ideal decoder, we can recover $\ccode \circ \elemid (\m)$ as desired.
\end{proof}

Next, we show that any constant-depth circuit is immediately $\epsilon$-robust for some constant $\epsilon$.

\begin{claim}
\label{claim:short-depth-robust}
    Let $\ccode: \bbF_2^{k} \to \bbF_2^{n}$ be a code such that for all $i\in [k]$, there exists $j \in [k]$, $j \neq i$, and a gadget $\elemft$ that corresponds to the $\CNOT$ gate between encoded bits $i,j \in [k]$:
    \begin{align*}
        \elemft \circ \ccode = \ccode \circ \CNOT_{i,j}~.
    \end{align*}
    Furthermore, suppose $\elemft$ can be implemented in depth $\ell$.
    Let $\epsilon$ such that $\epsilon = 2^{-\ell}$. 
    For all $E \subseteq [n]$ such that if $|E| < \epsilon d$, then
    \begin{align*}
        \elemft \circ \era_{E} \circ \ccode(\m) = \era_{F} \circ \elemft \circ \ccode(\m) ~.
    \end{align*}
\end{claim}
\begin{proof}
    Under the action of the circuit $\elemft$, an erasure of weight $|E|$ can evolve into an erasure of weight at most $2^{\ell}|E| < d$.
    Therefore, the operation results in the correct codeword.
\end{proof}

Claim~\ref{claim:short-depth-robust} then implies the following corollary.

\begin{corollary}
\label{cor:robust}
    Let $\famC$ be a $[n,k,d]$ family of codes that obey the constraints of Theorem~\ref{thm:main}.
    For $\ccode \in \famC$, where $\ccode$ is an $[n,k,d]$ code, let $\{\elemft_i\}_{i \in [k]}$ be the corresponding set of $\CNOT$ gates that can be implemented in depth at most $\ell$ where $\ell$ is independent of $n$.
    Then there exists a constant $\epsilon \in (0,1)$ such that for all elements $\ccode \in \famC$, all the $\CNOT$ gates $\{\elemft_i\}_i$ are $\epsilon$-robust where $\epsilon \leq 2^{-\ell}$.
\end{corollary}

\subsection{The Doubled Gadget}

In this section, we use the existence of a circuit $\elemft_i$ corresponding to the $\CNOT$ gate $\CNOT_{i,j}$ to construct a gadget that can be used to infer the $i$\textsuperscript{th} encoded message symbol $m_i$.

Let $\famC$ be a family of codes that obeys the constraints of Theorem~\ref{thm:main}.
Let $C : \bbF_2^{k} \to \bbF_2^{n}$ be an $[n,k,d]$ code in the family $\famC$.
Per the assumption in the statement of Theorem~\ref{thm:main}, for all $i \in [k]$, there exists a circuit $\elemft_{i}$ and at least one $j \in [k]$, $j \neq i$, such that the gadget $\elemft_{i}$ corresponds to the gate $\CNOT_{i,j}$.
Furthermore, this gate can be implemented in depth $\ell = O(1)$.

The encoded $\CNOT$ gates must be defined with respect to a basis $\g_1,...,\g_{k}$ for the code space.
For all $\m \in \bbF_2^{k}$, we can write $\ccode(\m) = \sum_{\alpha} m_{\alpha}\,\g_{\alpha}$.
The gadget $\elemft_i$ has the following action on codewords:
\begin{align}
\label{eq:encoded-cnot-wrt-basis}
   \sum_{\alpha} m_{\alpha}\,\g_{\alpha} \mapsto \sum_{\alpha} m_{\alpha}\, \g_{\alpha} + m_i \g_j~. 
\end{align}

Given a gadget $\elemft_i$, we define the doubled gadget $\cD(\elemft_i)$ that maps
\begin{align}
    (\ccode(\m), \bzero^n) \mapsto (\ccode(\m), m_i\, \g_j)~.
\end{align}
It is executed as follows, using the gadget $\elemft_i$ as a subroutine:
\begin{enumerate}
    \item Begin with $(\ccode(\m), \bzero^n)$.
    \item For all $\beta \in [n]$, perform the gate $\CNOT_{\beta, n + \beta}$.
    This results in $(\ccode(\m), \ccode(\m))$.
    \item Perform the gadget $\elemft_i$ on the second block of $n$ bits to obtain $(\ccode(\m), \ccode(\m) + m_i\,\g_j)$.
    This follows from the action of the gadget $\elemft$ as stated in Eq.~\eqref{eq:encoded-cnot-wrt-basis}.
    \item For all $\beta \in [n]$, perform the gate $\CNOT_{\beta, n + \beta}$.
    This results in $(\ccode(\m),m_i\,\g_j)$.
\end{enumerate}

The following claim follows immediately.
\begin{claim}
\label{claim:doubled-constant-depth}
    Suppose $\elemft_i$ can be implemented in depth $\ell$.
    The doubled gadget $\cD(\elemft_i)$ can be implemented in depth $\ell + 2$.
\end{claim}

The doubled gadget is useful because by querying the positions of $\supp(\g_j)$ in the locations $\{n+1,...,2n\}$, we can infer $m_i$.
Furthermore, as the circuit $\cD(\elemft)$ is a constant depth circuit, we can infer $m_i$ by querying $\ccode(\m)$ in a constant number of locations $q \leq 4\cdot2^{\ell}$.

\subsection{Influence}

In this section, we define the notion of influence.
Intuitively, it captures the number of input bits that that can affect an output bit.
The output bit is then sensitive to the value of these input bits.

Let $\elemft_{i}$ be a circuit implementing the encoded transformation $\elemid_i = \CNOT_{i,j}$, i.e.\ $\elemft_i \circ \ccode = \ccode \circ \elemid_i$.
Let $\double_{i} = \cD(\elemft_i)$ denote the doubled gadget constructed from $\elemft_i$.

\begin{definition}[Influence]
    Let $\cS_{j} = \supp(\g_{j})$ be the support of the $j$\textsuperscript{th} basis codeword.
    Given $\beta \in \cS_{j}$, the influence $\influence(\beta) \subseteq [n]$ is the set of indices such that there exists some path in $\double_{i}$ that terminates at $\beta$.
    This can naturally be extended to a set $\{\beta_1,...,\beta_{\Delta}\} \subseteq \cS_{j}$ of size $\Delta$, via the union: $\influence(\{\beta_1,...,\beta_\Delta\}) = \influence(\beta_1) \cup \cdots \cup \influence(\beta_\Delta)$.
\end{definition}

Let $i \in [k]$ and consider the $\CNOT$ gadget $\elemft_{i}$ that is guaranteed to exist per the statement of Theorem~\ref{thm:main}.
The influence $\influence(\elemft_i)$ of the circuit $\elemft_i$ is the maximum influence of the locations in it, i.e.\
\begin{align}
    \influence(\elemft_i) = \max_{\beta \in [n]}|\influence(\beta)|~.
\end{align}

\begin{lemma}
\label{lem:const-q}
    For all $\beta \in \cS_{j}$, $|\influence(\beta)| \leq 4 \cdot 2^{\ell}$.
\end{lemma}
\begin{proof}
     We restrict our attention to gadgets $\elemft_{i}$ with depth $\ell$ that are constructed using only gates with at most 2 inputs and 2 outputs.
     Consequently, from Claim~\ref{claim:doubled-constant-depth}, the doubled gadget $D_i$ has constant depth $\ell +2$.
    Hence, there can be at most a constant number of elements in $\influence(\beta)$ for all $\beta \in \cS_{j}$.
\end{proof}

\subsection{Robustness Implies Disjoint Influence Sets}
Our objective is to show that if the influence sets of the output bits overlap a lot, then the circuit $\elemft_{i}$ cannot be $\epsilon$-robust.
To be precise, we show that if there is a large amount of overlap in the influence sets, then an adversary can corrupt some small set of size $o(d)$ which will result in a corruption of encoded information.
However, this has to be done carefully as each $\beta \in \cS_{j}$ might infer $m_{i}$ in a robust way.
For instance, $\elemft_{i}$ may be able to perform some local error correction to overcome errors that the adversary introduces.

\begin{restatable}[]{lemma}{linearreln}
\label{linearreln}
    Let $i \in [k]$ such that there exists $j \in [k]$, $j \neq i$, and a gadget $\elemft_i$ corresponding to $\elemid_i = \CNOT_{i,j}$, i.e.\ $\elemft_i \circ \ccode = \ccode \circ \elemid_i$.
    Let $\cS_j = \supp(\g_j)$ be the support of the $j$\textsuperscript{th} codeword.
    For all $\beta \in \cS_{j}$, there exists a non-empty set of vectors $U^{i}_{\beta}$ such that for all $\u \in U^{i}_{\beta}$, $\supp(\u) \subseteq \influence(\beta)$ and
    \begin{align*}
        m_{i} = \langle \u, \ccode(\m) \rangle~.
    \end{align*}
    Furthermore, if for all $\u \in U^{i}_{\beta}$, there is an erasure in the support of $\u$, then we cannot infer $m_i$ from $\influence(\beta)$.
\end{restatable}

We defer the proof of this statement to Section~\ref{subsec:linearreln}.

This decomposition is useful because the circuit $\elemft_{i}$ might read $m_{i}$ in a robust way.
For example, it can flip $\beta \in \cS_j$ correctly even if there are errors in $\influence(\beta)$.

\begin{restatable}[]{lemma}{matchingexists}
\label{matchingexists}
    Fix $i \in [k]$ and let $\elemft_{i}$ be the gadget corresponding to $\elemid_i = \CNOT_{i,j}$.
    If it is robust to all errors $E$ of weight at most $\epsilon d$, then there exists a set of $U^{i}$ vectors such that:
    \begin{enumerate}
        \item For all $\u \in U^{i}$, $|\u| \leq 4 \cdot 2^{\ell}$.
        \item For all distinct pairs $\u,\u' \in U^{i}$, $\u$ and $\u'$ do not share support.
        \item For all $\u \in U^{i}$,
        \begin{align*}
            m_{i} = \langle \u, \ccode(\m) \rangle~.
        \end{align*}
    \end{enumerate}
    Furthermore, this set obeys $q = O(1)$ and $|U^{i}| = \Theta(d)$.
\end{restatable}
\begin{proof}
    Recall that $\cS_{j} = \supp(\g_j)$ and that Lemma~\ref{linearreln} guarantees that for all $\beta \in \cS_j$, there exists a  non-empty set of vectors $U^{i}_{\beta}$ such that for all $\u \in U^{i}_{\beta}$, $\supp(\u) \subseteq \influence(\beta)$ and that $m_i = \langle \u, \ccode(\m)\rangle$.
    
    Define a hypergraph $\cH = (\cV,\cE)$ whose vertices $\cV := [n]$ and for each $\beta \in \cS_{j}$ and each $\u \in U^{i}_{\beta}$, the support of $\u$ forms a hyperedge in $\cE$.
    From Lemma~\ref{linearreln}, each $\u$ is contained within $\influence(\beta)$ and from Lemma~\ref{lem:const-q}, contains at most $|\influence(\beta)| \leq 4 \cdot 2^{\ell}$ vertices.
    There are $\sum_{\beta} |U^{i}_{\beta}| = \Omega(d)$ hyperedges: This is because $|\cS_{j}| \geq d$ (as it corresponds to the support of a codeword $\g_{j}$), and $U^{i}_{\beta}$ is non empty.
    
    A matching is a set of hyperedges such that no two hyperedges share a vertex.
    A \emph{maximum} matching is a matching with the largest possible number of hyperedges.
    A matching is \emph{maximal} if no more edges can be added to it.
    A maximum matching is maximal (but a maximal matching need not be maximum).
    
    We shall let $\cU^{i}$ be the maximum matching on the hypergraph $\cH$.
    For the sake of contradiction, suppose $\cU^{i}$ has size $o(d)$.
    Let $\cV(\cU^{i}) \subseteq \cV$ be the vertices that are covered by edges in $\cU^{i}$.
    The cardinality of $\cV(\cU^{i})$ is also $o(d)$ as every hyperedge has constant degree.
    
    The set $\cE \setminus \cU^{i}$ are the edges that are not included in the maximum matching.
    Every edge in $\cE \setminus \cU^{i}$ is incident to at least one vertex in $\cV(\cU^{i})$ (because $\cU^{i}$ is a maximal matching).
    By corrupting $\cV(\cU^{i})$, we corrupt all equations in $\cE$.
    This is because of Lemma~\ref{linearreln}: For $\beta \in \cS_{j}$, if for all $\u \in U^{i}_{\beta}$, there is an erasure in the support of $\u$, then we cannot read $m_i$ from $\influence(\beta)$.
    This then guarantees that we can corrupt $d$ of the output variables using only $|\cV(\cU^{i})| = o(d)$ errors.
    In particular, these $d$ corruptions correspond to the support of a codeword and therefore will corrupt encoded information.
    However, this is a contradiction as we showed that the encoded $\CNOT_{i,j}$ is robust to $E = \epsilon d$ errors for some constant $\epsilon$ in Corollary~\ref{cor:robust}.
    
    The maximum matching $\cU^{i}$ defines the set of vectors $U^{i}$.
\end{proof}

In the next lemma, we justify the assumption that there must be a large subset of vectors such that we cannot infer $m_i$ using a single query.

\begin{lemma}
\label{lem:lower-bound-inf-1}
    Suppose we have a code family $\ccode$ such that $kd = \omega(n)$.
    For all $i \in [k]$, let $U^{i}$ be the set of vectors guaranteed to exist by Lemma~\ref{matchingexists} that allow one to infer $m_i$.
    We can assume that for all $\u \in U^{i}$, we have $|\u| > 1$.
\end{lemma}
\begin{proof}
    From Lemma~\ref{matchingexists}, for all $i \in [k]$, we have $|U^{i}| = \Omega(d)$.
    Suppose there exists $i \in [k]$ such that $|\u| = 1$ for all $\u \in U^{i}_{\beta}$.
    However, this would mean that there are $\Omega(d)$ ways to access $m_i$ using just a single query.
    If this is possible for all $\Omega(k)$ gates, then it must be the case that $kd = O(n)$.
    This contradicts the assumption that $kd = \omega(n)$.

    By discarding some message symbols if need be, we assume that for all $i \in [k]$, for all $\u \in U^{i}$, $|\u| > 1$.
\end{proof}

\subsection{Proof of Lemma \ref{linearreln}}
\label{subsec:linearreln}
In this section, we prove Lemma \ref{linearreln}.
For convenience, we restate the lemma here before proceeding to the proof.
\linearreln*
\begin{proof}
Let $\beta \in \cS_j$ and observe that it is flipped if and only if $m_i = 1$.
Let $\influence(\beta) \subseteq [n]$ denote the influence of $\beta$.
For the sake of contradiction, suppose there are no linear functions in $\influence(\beta)$ used to ascertain whether $\beta$ will be flipped.

Let $\mathsf{Pow}(\influence(\beta))$ be the power set of $\influence(\beta)$.
Any Boolean function from $\influence(\beta)$ that outputs a single bit $m_i$ can be expressed as
\begin{align}
\label{eq:contra-nonlinear}
    m_{i} = \sum_{I \in \mathsf{Pow}(\influence(\beta))} m_{I} \ccode(\m)^{I}~,
\end{align}
for some coefficients $m_{I} \in \{0,1\}$ where $\ccode(\m)^{I} = \prod_{\alpha \in I} \ccode(\m)_{\alpha}$ where $\ccode(\m)_{\alpha}$ represents the $\alpha$\textsuperscript{th} component of $\ccode(\m)$.
By the assumption we made towards contradiction, the degree of the polynomial $\sum_{I} m_{I} \ccode(\m)^{I}$ is $\delta$.
Consequently, for some $J \in \mathsf{Pow}(\influence(\beta))$, $|J| = \delta > 1$ such that $m_{J} = 1$.

Consider the polynomial $f$ defined as
\begin{align}
    f(\m) = \sum_{I} m_{I} \ccode(\m)^{I} + m_{i}~.
\end{align}
As $\ccode(\m)$ and $\m$ are related linearly (via the generator matrix), $f$ is also a non-trivial polynomial over $\{0,1\}$ that has degree $\delta > 1$.
Hence, it cannot be identically zero.
This contradicts the assumption in Eq. \eqref{eq:contra-nonlinear}.

Thus, there must be linear functions over $\influence(\beta)$ that can evaluate $m_i$.
\end{proof}

\subsection{Proof of Theorem \ref{thm:main}}

In the previous sections, we have stated and proved all the tools required to prove our main result, Theorem~\ref{thm:main}.
We conclude the paper by assembling these results to prove Theorem~\ref{thm:main}.

\begin{proof}(Theorem \ref{thm:main}).
    Combining Lemma~\ref{matchingexists} and Lemma~\ref{lem:lower-bound-inf-1}, we have shown that the code family $\famC$ satisfies Def.~\ref{def:locality} and, therefore, is a family of codes with $(q,r)$ locality where $1 < q \leq 4 \cdot 2^{\ell}$ and $r = \Theta(d)$.
    As the depth $\ell$ of the encoded $\CNOT$ gates is assumed to be a constant independent of the block size $n$, we must have $q = O(1)$.

    Therefore, we can use the bound on the code dimension from Theorem~\ref{thm:bound-ldc}: The dimension $k$ of the code must be sublinear in $n$, obeying
    \begin{align}
        k = O\left(\frac{n}{d^{1/q}}\right)~.
    \end{align}
    This completes the proof of Theorem~\ref{thm:main}.
\end{proof}

\pagebreak

\bibliographystyle{alpha}
\bibliography{references}

\appendix

\section{Explicit schemes for fault-tolerant classical computation}
\label{app:explicit}

Our main result shows that to construct explicit schemes for performing fault-tolerant $\CNOT$ gates, a code family must correspond to a $(q,r)$ local code.
In this section, we go the other direction: Beginning with a $(q,r)$ local code, we present explicit fault-tolerant circuits for performing encoded gates.
This is only included for the sake of completeness.
As noted in the main text, the work of Romashchenko \cite{romashchenko2006reliable} has already presented a scheme for fault-tolerant computation using locally decodable codes.

Let $\famC = \{\ccode: \bbF_2^{k} \to \bbF_2^{n}\}$ be an $[n,k(n),d(n)]$ code family such that it is a $(q,r)$ local code, where $q$ is some constant and $r = \gamma d$ where $\gamma \in (0,1)$ is some constant.

Let $\cid$ be an circuit with width $\widthid$ and depth $\depthid$.
We pick $\ccode \in \famC$ such that it has the smallest $n$ for which $k \geq \widthid$.
Let $\{\g_{\alpha}\}_{\alpha \in [k]}$ denote the privileged basis for $\ccode$.
For simplicity, we assume that $d = \Theta(n)$.
Furthermore, we assume that all codewords in the privileged basis have length $\Theta(d)$, as otherwise, the encoded $\CNOT$ gate may not implementable in constant depth.

We assume that $\cid$ is sparse, i.e.\ at every time step, there is one and only one encoded $\CNOT$ gate being implemented.
We shall specify a protocol to perform the encoded $\CNOT$ gate between registers $i,j \in [k]$.

\textbf{Circuit:}
\begin{enumerate}
    \item Divide the support of $\cS_{j} := \supp(\g_{j})$ arbitrarily into $r$ disjoint subsets $\cS^{1}_{j},...,\cS^{r}_{j} \subseteq [n]$ such that for all $\alpha \in r$, $\cS^{\alpha}_{j}$ has size at most $|\g_{j}|/r$.
    \item As $\ccode$ is a $(q,r)$ local code, there exist vectors $\v^{i}_1,...,\v^{i}_r$ to infer $m_i$.
    For $\alpha \in [r]$:
    \begin{enumerate}
        \item \textbf{Inference:} We infer $\hat{m}_{i,\alpha} = \langle \v^{i}_{\alpha}, \ccode(\m)\rangle$.
        We have
        \begin{align}
            \hat{m}_{i,\alpha} =
            \begin{cases}
                m_i &\mbox{if there are no erasures in } \supp(\v^{i}_{\alpha})\\
                \perp &\mbox{otherwise}
            \end{cases}~.
        \end{align}
        \item \textbf{Flip:} If $\hat{m}_{i,\alpha} = 1$, we flip all bits in $\cS^{\alpha}_{j}$ .
        However, if $\hat{m}_{i,\alpha} = \perp$, we erase all locations in $\cS^{\alpha}_{j}$.
    \end{enumerate}
\end{enumerate}

The following claim shows that this gate is robust to $\epsilon d$ erasure errors for some constant $\epsilon$.

\begin{claim}
\label{claim:app-fault-tolerant}
    This scheme is fault tolerant to up to $\epsilon d$ erasure errors where $\epsilon = \min_{j} \frac{r}{r+|\g_j|}$.
\end{claim}
\begin{proof}
    Let $i,j \in [k]$ be distinct and consider the protocol for performing the encoded $\CNOT_{i,j}$ gate above.
    For $\beta \in [r]$, we can corrupt the inference made by $\v^{i}_{\alpha}$ by corrupting a single variable in $\supp(\v^{i}_{\alpha})$.

    Let $\eta_{j} = |\g_{j}|/r$ be a constant.
    Each erasure in the input can result in at most $1+\eta$ erasure errors in the final word.

    If the adversary can corrupt $\epsilon d$ locations, then the output word will have strictly less than $d$ erasure errors if
    \begin{align}
        \left(1 + \eta_{j} \right)\epsilon < 1~.
    \end{align}
    The claim follows by minimizing $\epsilon$ over all possible $\CNOT$ gates.
\end{proof}

We conclude this section by noting that the circuit above used to implement encoded $\CNOT$ can tolerate $\epsilon$ fraction of errors where $\epsilon$ is a constant.

\begin{corollary}
    Let $\ccode : \bbF_2^{k} \to \bbF_2^{n}$ be a $(q,r)$ local code such that $q = O(1)$ and $r = \gamma d$ for some constant $\gamma \in (0,1)$.
    Let $B_{\ccode} = \{\g_1,...,\g_k\}$ be the privileged basis for $\ccode$ and assume that for all $i \in [k]$, that $|\g_i| = \Theta(d)$.
    Then the fraction of errors $\epsilon$ that the encoded $\CNOT$ is robust to is a constant.
\end{corollary}
\begin{proof}
    Let $\zeta = d^{-1} \cdot \max_{j} |\g_j|$.
    According to Claim~\ref{claim:app-fault-tolerant}, we have
    \begin{align}
        \epsilon = \min_{j} \frac{r}{r + |g_j|} \geq \frac{\gamma}{\gamma + \zeta}~.
    \end{align}
    In the inequality, we have substituted $r = \gamma d$.
    This completes the proof.
\end{proof}

\end{document}